\documentclass[a4paper,runningheads,11pt]{llncs}
\usepackage{hyperref}
\usepackage{amsfonts,amssymb,amscd}
\usepackage{latexsym, graphicx}
\usepackage{tabularx}
\usepackage{multirow}

\usepackage[margin=3cm]{geometry}

\graphicspath{{figures/}}

\usepackage[ruled,noend,nofillcomment,linesnumbered]{algorithm2e}
\SetAlTitleFnt{}
\SetAlgoInsideSkip{bigskip}
\SetKwBlock{Begin}{begin}{end}
\SetArgSty{}
\SetCommentSty{textit}

\usepackage{enumitem}

\newtheorem{observation}[lemma]{Observation}

\renewcommand{\claim}[1]{\smallskip\par\noindent\textit{Claim #1}}
\newenvironment{claimproof}{%
\par%
}
{
\hfill $\triangle$%
\smallskip\par
}

\newcommand{\cNP}{\hbox{\textsf{NP}}}

\begin{document}

\title{Intersection Graphs of Non-Crossing Paths\thanks{A preliminary version of this article appeared at WG 2019~\cite{wg2019}, and a preprint is available at~\href
{http://arxiv.org/abs/1907.00272}{arxiv.org/abs/1907.00272}}}

\author{Steven Chaplick\orcidID{0000-0003-3501-4608}
\thanks{Part of this research was conducted while the author was employed at Lehrstuhl f\"ur Informatik~I, Universit\"at W\"urzburg, and partially supported by DFG grant WO$\,$758/11-1.}
}

	\institute{Department of Data Science and Knowledge Engineering, Maastricht University, The Netherlands\\
		\email{s.chaplick@maastrichtuniversity.nl}}

\maketitle

\begin{abstract}
We study graph classes modeled by families of non-crossing (NC) connected sets. 
Two classic graph classes in this context are disk graphs and proper interval graphs. 
We focus on the cases when the sets are paths and the host is a tree (generalizing proper interval graphs). 
Forbidden induced subgraph characterizations and linear time certifying recognition algorithms are given for intersection graphs of NC paths of a tree (and related subclasses).
A direct consequence of our certifying algorithms is a linear time algorithm certifying the presence/absence of an induced claw $(K_{1,3})$ in a chordal graph. 

For the intersection graphs of NC paths of a tree, we characterize the minimum connected dominating sets (leading to a linear time algorithm to compute one). 
We further observe that there is always an independent dominating set which is a minimum dominating set, leading to 
the dominating set problem being solvable in linear time. 
Finally, each such graph $G$ is shown to have a Hamiltonian cycle if and only if it is 2-connected, and when $G$ is not 2-connected, a minimum-leaf spanning tree of $G$ has $\ell$ leaves if and only if $G$'s block-cutpoint tree has exactly $\ell$ leaves (e.g., implying that the block-cutpoint tree is a path if and only if the graph has a Hamiltonian path).
\end{abstract}

\keywords{Clique Trees \and Non-crossing Models \and Dominating Sets \and Hamiltonian Cycles \and Minimum-Leaf Spanning Trees.}

\section{Introduction}

Intersection models of graphs are ubiquitous in graph theory and covered in many graph theory textbooks, see, e.g., \cite{Golumbic2004,MckeeM1999}. 
Generally, for a given graph $G$ with vertex set $V(G)$ and edge set $E(G)$, a collection $\mathcal{S}$ of sets, $\{S_v\}_{v \in V(G)}$, is an \emph{intersection model} of $G$ when $S_u \cap S_v \neq \emptyset$ if and only if $uv \in E(G)$. 
Similarly, we say that $G$ is the \emph{intersection graph} of~$\mathcal{S}$. 
One quickly sees that all graphs have intersection models (e.g., by choosing, for every $v \in V(G)$, $S_v$ to be the edges incident to $v$). 
Thus, one often considers restrictions either on the \emph{host} set (i.e., the domain from which the elements of the $S_v$'s can be chosen), collection $\mathcal{S}$, and/or on the individual sets $S_v$. 

In this paper we consider classes of intersection graphs where the sets are taken from a topological space, are \emph{(path) connected}, and are pairwise \emph{non-crossing}. 
A set $S$ is \emph{(path) connected} when any two of its points can be connected by a \emph{curve} within the set (note: a \emph{curve} is a homeomorphic image of a closed interval). 
Notice that, when the topological space is a graph, connectedness is precisely the usual connectedness of a graph and curves are precisely paths.  
Two connected sets $S_1,S_2$ are called \emph{non-crossing} when both $S_1 \setminus S_2$ and $S_2 \setminus S_1$ are connected. 
Our focus will be on intersection graphs of non-crossing paths. %

The most general case of intersection graphs of non-crossing sets which has been studied is the class of intersection graphs of non-crossing connected (NC-C) sets in the plane~\cite{Kratochvil1997}. 
These were considered together with another non-crossing class, the intersection graphs of disks in the plane or simply \emph{disk} graphs. 
The recognition of both NC-C graphs and disk graphs is \cNP-hard~\cite{Kratochvil1997}. 
More recently~\cite{Kang2012}, disk graph recognition was shown to complete for the \emph{existential theory of the reals} ($\exists\mathbb{R}$); note that all $\exists\mathbb{R}$-hard problems are \cNP-hard, see~\cite{Matousek2014} for an introduction to $\exists\mathbb{R}$.

One of the simplest cases of connected sets one can consider are those which reside in $\mathbb{R}$, i.e., the intervals of $\mathbb{R}$. 
The corresponding intersection graphs are precisely the well studied \emph{interval graphs}. 
Moreover, imposing the non-crossing property on these intervals leads to the \emph{proper} interval graphs -- which are usually (and equivalently) defined by restricting the guests intervals so that no interval strictly contain any other.
It has often been considered how to generalize proper interval graphs to more complicated hosts, but  
 simple attempts to do so involving the property that the sets are \emph{proper} are often uninteresting. 
For example, the intersection graphs of proper paths in trees or proper subtrees of a tree are easily seen as the same as their non-proper versions. 
We will see that the non-crossing property leads to natural new classes which generalize proper interval graphs. 

We formalize the setting as follows. For graph classes $\mathcal{S}$ and $\mathcal{H}$, a graph $G$ is an \emph{$\mathcal{S}$-$\mathcal{H}$} graph when each $v \in V(G)$ has an $S_v \in \mathcal{S}$ such that:
\begin{itemize}[noitemsep,topsep=0.5\baselineskip]
\item the graph $H = \bigcup_{v \in V(G)} S_v$ is in $\mathcal{H}$, and 
\item $uv$ is an edge of $G$ if and only if $S_u \cap S_v \neq \emptyset$. %
\end{itemize} 
Additionally, we say that $(\{S_v\}_{v \in V(G)},H)$ is an \emph{$\mathcal{S}$-$\mathcal{H}$ model} of $G$ where $H$ is the \emph{host} and each $S_v$ is a \emph{guest}, we will also refer to $S_v$ as the \emph{model of $v$}. 
We further state that $G$ is a \emph{non-crossing}-$\mathcal{S}$-$\mathcal{H}$ (NC-$\mathcal{S}$-$\mathcal{H}$) graph when the sets $S_v$ are pairwise non-crossing. 
In this context the proper interval graphs are the NC-path-path graphs.

Many classes of $\mathcal{S}$-$\mathcal{H}$ graphs have been studied in the literature; see, e.g.,~\cite{MckeeM1999}. 
Some of these are described in the table below together with the complexity of their recognition problems and whether a \emph{forbidden induced subgraph characterization (FISC)} is known. 
The table utilizes the following terminology. 
A \emph{directed tree (d.tree)} is a tree in which every edge $uv$ has been assigned one direction.
A \emph{rooted tree (r.tree)} is a directed tree where there is exactly one source node. 
A survey of path-tree graph classes is given in \cite{MonmaW86}.

Two further key graph classes here are the chordal graphs and the split graphs, defined as follows.
A graph is \emph{chordal} when it has no induced cycles of length four or more.
A graph is a \emph{split} graph when its vertices can be partitioned into a clique and an independent set. 
The split graphs are easily seen as a subset of the chordal graphs.

\begin{center}
\begin{tabular}{@{~~~}l@{~~~}l@{~~~}l@{~~~}l@{~~~}l@{~~~}l}
\hline
& Graph Class & Guest & Host~~~~ & Recognition & FISC? \\
\hline 
1 & interval & path & path & $O(n+m)$ \cite{CorneilOS2009} & yes \cite{LB1962} \\
2 & rooted path tree (RPT) & path & r.tree & $O(n+m)$ \cite{Dietz1984} & open \\
3 & directed path tree (DPT) & path & d.tree & $O(nm)$ \cite{ChaplickGLT2010} & yes \cite{Panda1999}  \\
4 & path tree (PT) & path & tree & $O(nm)$ \cite{Schaffer1993} & yes \cite{LevequeMP2009} \\
5 & chordal & tree & tree & $O(n+m)$ \cite{RoseTL1976} & by definition \\
\hline 
\end{tabular}
\end{center}

\paragraph{Results and outline.}
We study the non-crossing graph classes corresponding to graph classes 1--4 given in the table. 
Section~\ref{sec:prelim} contains background, terminology, and notation concerning intersection models.

In Section~\ref{sec:nc-path}, we provide forbidden induced subgraph characterizations for the non-crossing classes corresponding to 1--4 and certifying linear time recognition algorithms for them. 
Interestingly, this implies that one can test whether a chordal graph contains a claw in linear time. %
 (In contrast, for general graphs, the best deterministic claw-detection algorithms run in time $O(\min\{n^{3.252},m^{(\omega+1)/2}\})$~\cite{EisenbrandG04}, and $O(m^{\frac{2\omega}{\omega+1}})$~\cite{KloksKM00}, whereas the best randomized algorithm (succeeding with high probability) runs in time $O(n^\omega)$~\cite{WilliamsWWY15}; here $\omega$ is the exponent from square matrix multiplication and the $3.252$ is based on the time to compute the product of an $n \times n^2$ matrix and an $n^2 \times n$ matrix~\cite{GallU18}.)

The next two sections concern algorithmic results on domination and Hamiltonicity problems on NC-path-tree graphs. To obtain these results, we use the special structure of NC-path-tree models established in Section~\ref{sec:nc-path-tree-structure}. 
Note that the problems mentioned below are formalized in the corresponding sections.

In Section~\ref{sec:mds}, our main result is a characterization of the minimum connected dominating sets in NC-path-tree graphs. 
This leads to a linear time algorithm to solve the minimum connected dominating set problem, which also implies a linear time algorithm for the cardinality Steiner tree (ST) problem. 
In contrast, the minimum connected dominating set problem is known to be \cNP-hard on split graphs~\cite{WhiteFP85}, and, as such, on chordal graphs as well.
We further discuss the relationship between the (standard) minimum dominating set problem and the minimum independent dominating set problem on these graphs, observing that both can be solved in linear time on NC-path-tree graphs. 
Notably, the minimum dominating set problem is \cNP-hard on PT graphs~\cite{BoothJ1982}, and split graphs~\cite{CorneilP1984}, 
but it is polynomial time solvable on RPT graphs~\cite{BoothJ1982}. 
Further references and background on domination problems are given in Section~\ref{sec:mds}. 

In Section~\ref{sec:ham}, we again consider NC-path-tree graphs but now study the Hamiltonian cycle (HC) problem and minimum-leaf spanning tree problem (which generalizes the Hamiltonian path (HP) problem), on them. 
We show that 2-connectedness implies that each plane drawing of an NC-path-tree model leads to a distinct HC, and from each such plane drawing, an HC can be found in linear time. 
When such a graph is not 2-connected (i.e., when it has a cut-vertex) it cannot have an HC, but we can similarly observe a nice spanning substructure.
Namely, we show that for any NC-path-tree graph $G$ containing a cut-vertex, its block-cutpoint tree having $\ell$ leaves characterizes the presence of an $\ell$-leaf spanning tree in $G$. 
For example, as a special case, we obtain that $G$'s block-cutpoint tree is a path (i.e., has at most two leaves) if and only if $G$ has an HP.
Our characterization also leads to a linear time algorithm for the minimum-leaf spanning tree problem. 
Note that, the HC and HP problems are NP-complete on \emph{strongly chordal} split graphs~\cite{Muller96a}, and DPT graphs~\cite{Narasimhan89}, but easily solved (and similarly characterized) on proper interval graphs~\cite{Bertossi83}. 

We conclude with avenues for further research.

\section{Preliminaries}
\label{sec:prelim}
\paragraph{Notation.}
Unless explicitly stated otherwise, all the graphs we discuss in this work are connected, undirected, simple, and loopless. 
For a graph $G$ with a vertex $v$, we use $N_G(v)$ to denote the \emph{neighborhood} of $v$, and $N_G[v]$ to denote the \emph{closed neighborhood} of $v$, i.e., $N_G[v] = N_G(v) \cup \{v\}$. 
The subscript $G$ will be omitted when it is clear. %
For a subset $S$ of $V(G)$, we use $G[S]$ to denote the subgraph of $G$ induced by $S$. 
For a set of graphs $\mathcal{F}$, we say that a graph $G$ is $\mathcal{F}$-free when $G$ does not contain any $F \in \mathcal{F}$ as an induced subgraph. 

For graph classes $\mathcal{S}$ and $\mathcal{H}$, and an $\mathcal{S}$-$\mathcal{H}$ model $(\{S_v\}_{v \in V(G)}, H)$ of a graph $G$, we use the following notation. 
We refer to elements of $V(G)$ as \emph{vertices} and use symbols $u$ and $v$ to refer them whereas we call elements of $V(H)$ \emph{nodes} and use $x$, $y$, and $z$ to refer to them. 
For a node $x$ of $H$ we use $G_x$ to denote the set of vertices $v$ in $G$ where $S_v$ contains $x$. 
Observe that every set $G_x$ induces a clique in $G$. 
Note that Section~\ref{sec:nc-path-tree-structure} defines the terms \emph{terminal}, \emph{junction}, and \emph{mixed} that are also used in later sections of the paper.

Several special graphs are named and depicted in Fig.~\ref{fig:small-graphs} along with models of them. 
We will refer to these throughout this paper. 
Of particular note is the middle graph, $K_{1,3}$ aka the claw, where we will refer to its degree~3 vertex as its \emph{central vertex}. 

\begin{figure}[tb]
\centering
\begin{tabular}{ccp{0.7cm}cccp{0.7cm}cc}
\includegraphics[scale=1]{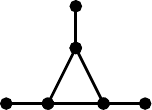} & \includegraphics[scale=1]{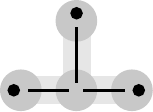} & & \includegraphics[scale=1]{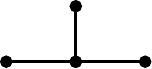} & \includegraphics[scale=1]{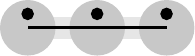}  & \includegraphics[scale=1]{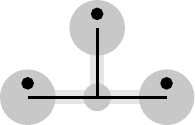} & & \includegraphics[scale=1]{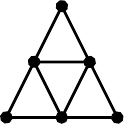} & \includegraphics[scale=0.8]{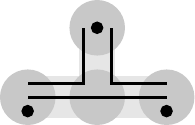} \\
\multicolumn{2}{l}{\quad ~net} & & \multicolumn{3}{l}{\quad claw ($K_{1,3}$)} & & \multicolumn{2}{l}{~3-sun} \\
\end{tabular}

\caption{Some small graphs and tree-tree models of them. 
In the models the nodes of the host graph are given as darkly shaded circles and its edges are lightly shaded corridors connecting them. 
Each subset $S_v$ is depicted by a tree (or single point) overlaid on the drawing of the host graph. }%
\label{fig:small-graphs}
\end{figure}

\paragraph*{Twin-free Graphs.}
For a graph $G$, two vertices $x$ and $y$ are called \emph{twins} when they have the same closed neighborhood, i.e., $N[x] = N[y]$. 
Note that, for the MDS problem, it is an easy exercise to show that it suffices to consider twin-free graphs. 
Also, as the vertex set of a graph can be easily partitioned into its equivalence classes of twins in linear time, one can distill the relevant twin-free induced subgraph of $G$ in linear time. 

\paragraph{Chordality and Clique Trees.}\label{sec:chordal}
This area is deeply studied and while there are many interesting results related to our work, we only pick out a few concepts and results which are useful in this paper. 
The starting point is that the chordal graphs are well-known to be the tree-tree graphs~\cite{Buneman74,Gavril1974,Walter1978}. 

For a chordal graph $G$, a \emph{clique tree} $T$ of $G$ has the maximal cliques of $G$ as its vertices, and for every vertex $v$ of $G$, the set $K_v$ of maximal cliques containing $v$ induces a subtree of $T$. 
In other words, a clique tree of $G$ is a tree-tree model of $G$ whose nodes are in bijection with the maximal cliques of $G$. 
Clique trees are very useful when discussing models where the host graph is a tree. 
When a graph has a tree-tree~\cite{Buneman74,Gavril1974,Walter1978}, path-tree~\cite{Gavril1978}, path-d.tree~\cite{MonmaW86}, path-r.tree~\cite{Gavril1975}, or path-path~\cite{FG1965} model, then it also has one that is a clique tree. Such results are also summarized in~\cite{MckeeM1999}.

We establish similar clique tree results for the corresponding NC graphs when the guests are paths. 
However, we remark that when the guests are trees, we cannot rely on clique trees. 
For example, the claw ($K_{1,3}$) is an NC-tree-tree graph, but it does not have an NC-tree-tree model that is a clique tree; in particular, in the center of Fig.~\ref{fig:small-graphs}, we depict two tree-tree models of the claw: one is the only clique tree (which is readily seen to fail the NC condition due to the point in the ``middle'' node) and the other is an NC-tree-tree model. 

Essential to the linear running time of our algorithms is the following property of maximal cliques of chordal graphs, and ultimately clique trees. 
For a chordal graph $G$, $\sum_{v\in V(G)} |K_v| \in O(n+m)$~\cite{Golumbic2004}. 
This implies that the total size of a clique tree $T$ is $O(n+m)$. 
So, any algorithm that is linear in the size of $T$ is also linear in the size of $G$.
Additionally, one can produce a clique tree of a chordal graph in linear time~\cite{BlairP1993,GalinierHP1995}

One final aspect of clique trees which is relevant for us is the study of chordal graphs with unique clique trees~\cite{Kumar2002}. 
One observation that they made is that claw-free chordal graphs have unique clique trees. 
This is relevant for us because, as we will see in Section~\ref{sec:nc-path}, the claw-free chordal graphs are precisely the NC-path-tree graphs. 
The uniqueness of the clique trees of proper interval graphs (a subclass of claw-free chordal graphs) was also observed (later) in~\cite{Ibarra2009}.
In fact, a very recent paper~\cite{Grussien2019} specifically studies the claw-free chordal graphs, providing a logarithmic-space isomorphism test while also re-proving that claw-free chordal graphs have unique clique trees. 
\section{Non-crossing Paths in Trees: Structure and Recognition}
\label{sec:nc-path}

In this section we characterize and recognize classes of intersection graphs of non-crossing paths in trees; namely, the classes of the following types: NC-path-tree, NC-path-d.tree, NC-path-r.tree, and NC-path-path. 
We first note that the \emph{claw} ($K_{1,3}$) is not an NC-path graph regardless of the host. 

\begin{observation}\label{obs:nc-path=>claw-free}
If $G$ is an NC-path graph, then $G$ is claw-free. 
\end{observation}
\begin{proof}
Suppose $G$ contains a claw with central vertex $u$ and pendant vertices $a$, $b$, $c$. 
Let $\mathcal{P}$ be a path-$\mathcal{H}$ model of $G$ where $\mathcal{P} = \{P_v\}_{v \in V(G)}$. Clearly, $P_a \cap P_u$, $P_b \cap P_u$ and $P_c \cap P_u$ are disjoint. As such, at most two of them include an endpoint of $P_u$. Thus, for some $d\in \{a,b,c\}$, $P_u \setminus P_d$ is disconnected. 
\end{proof}

This section proceeds as follows. 
The NC-path-tree graphs are shown to be the claw-free chordal graphs and the structure of NC-path-tree models is described.
From this structure, we then show that NC-path-d.tree $=$ NC-path-r.tree $=$ (claw,3-sun)-free chordal. 
This provides, as a nearly direct consequence, the classic result that proper interval graphs are precisely the (claw, 3-sun, net)-free chordal graphs~\cite{Roberts1970,Wegner1967}.
We conclude with linear time certifying recognition algorithms for NC-path-tree and NC-path-r.tree graphs.

\subsection{The Structure of NC-path-tree Models}
\label{sec:nc-path-tree-structure}

In this subsection we explore the structure of NC-path-tree models and prove our FISCs along the way. 
We first take a slight detour to claw-free chordal graphs and prove the FISC of NC-path-tree graphs. 
In doing so we obtain the first insight into NC-path-tree models. 
Namely, that it suffices to consider clique trees and that the clique trees of these graphs are unique (see Theorem~\ref{thm:nc-path-tree-fisc}). 
We then take a closer examination of these clique NC-path-tree models and carefully describe the nodes they contain -- the results of this examination will be used repeatedly in the rest of the paper. %

\begin{theorem}
\label{thm:nc-path-tree-fisc}
A graph $G$ is claw-free chordal if and only if it is an NC-path-tree graph. Moreover, $G$ has a unique clique tree and this clique tree is an NC-path-tree model. 
\end{theorem}
\begin{proof}%
\noindent $\Leftarrow$ Observation~\ref{obs:nc-path=>claw-free} and chordal graphs being tree-tree graphs imply this.

\noindent $\Rightarrow$ 
Let $T$ be a clique tree of a claw-free chordal graph $G$. 
In the two claims below, we first show that every subtree $T_v$ must be a path, and then we show that these paths are non-crossing. 
These two claims prove the characterization. 
The uniqueness of the clique tree of every claw-free chordal graph has been shown previously~\cite{Kumar2002}.

\claim{1: For every $v \in V(G)$, $T_v$ is a path.}
\begin{claimproof}
Suppose $T_v$ is not a path. Then $T_v$ contains some claw $x_0,x_1,x_2,x_3$ with central node $x_0$. However, since $G_{x_j}$ is a maximal clique (for each $j \in \{0,1,2,3\}$), for each $i \in \{1,2,3\}$, there is $v_i \in G_{x_i} \setminus G_{x_0}$. Thus $v, v_1, v_2, v_3$ induces a claw in $G$. 
\end{claimproof}

\claim{2: The set $\{T_v : v \in V(G)\}$ is non-crossing.}
\begin{claimproof}
Suppose that $T_u$ intersects $T_v$ but does not include either end of $T_v$. Let $x_1$ and $x_2$ be the endpoints of $T_v$. 
Now there must be $v_1 \in G_{x_1} \setminus N_G(u)$ and $v_2 \in G_{x_2} \setminus N_G(u)$. That is, $v, u, v_1, v_2$ induces a claw in $G$. 
\end{claimproof}
\end{proof}

We now study the structure of the clique NC-path-tree model $(\{P_v\}_{v\in V(G)},T)$ of a graph $G$. 
We introduce some terminology. 
A node $x$ of $T$ is called a \emph{terminal} when it is a leaf of every path which contains it, i.e., $x$ is not an internal node of any $P_v$. 
For example, the leaves of $T$ are terminals. 
Similarly, a node $x$ of $T$ is a \emph{junction} when it is an internal node of every path which contains it, i.e., $x$ is not a leaf of any $P_v$. 
A node of $T$ that is neither a terminal nor a junction is called \emph{mixed}. 
The remainder of this section consists~of
\begin{itemize}
 \item the main lemma describing $T$ in these terms (Lemma~\ref{lem:nc-path-tree.nodes}),
 \item  an observation connecting these terms with certain induced subgraphs of $G$ (Observation~\ref{obs:path-tree.nets.suns}), and
 \item a corollary regarding how the terminals can be used to partition $T$ into ``simple'' subtrees (Corollary~\ref{cor:nc-path-tree.structure}).
\end{itemize}

\begin{lemma}\label{lem:nc-path-tree.nodes}
For an NC-path-tree graph $G$, let $(\{P_v\}_{v\in V(G)},T)$ be its clique NC-path-tree model. 
A node $x$ of $T$ must satisfy the following properties: 
\begin{enumerate}[noitemsep,topsep=0pt]
\item If $x$ is mixed, then $x$ has degree two. 
\item If $x$ is a junction, then (i) $x$ has degree 3 and (ii) $x$'s neighbors are terminals.
\item If $x$ has degree four or more, then $x$ is a terminal.
\end{enumerate}
\end{lemma}
\begin{proof} We establish the claimed properties in order as follows.

\begin{claimproof}
\textbf{1.:} Suppose that $x$ has degree at least 3, is a leaf of $P_v$, and is an internal node of $P_u$. 
Further, let $y$ be the unique neighbor of $x$ in $P_v$. 
We see that $P_u$ includes $y$ (otherwise, $P_v$ and $P_u$ cross). 
Let $y'$ be the neighbor of $x$ in $P_u \setminus P_v$ and let $y''$ be a neighbor of $x$ which is not in $P_u$. 
Since $G$ is connected, there exists $u' \in G_x \cap G_{y''}$. 
Furthermore, $x$ is not a leaf of $P_{u'}$ (otherwise, $P_u$ crosses $P_{u'}$). 
Thus, similarly to $P_u$, $y$ belongs to $P_{u'}$. 
Now, since $G_x$ and $G_{y}$ are maximal cliques, there is $u'' \in G_x \setminus G_{y}$. 
Thus for $P_{u''}$ to neither cross $P_{u}$ nor $P_{u'}$ it must include both $y'$ and $y''$. 
However, this means $P_{u''}$ and $P_v$ cross. 
\end{claimproof}

\begin{claimproof}
\textbf{2.:} Suppose that $x$ is a junction and let $y_1, \ldots, y_k$ be the neighbors of $x$. 
Since $x$ is a junction, for every $v \in G_x$, $P_v$ contains exactly two $y_i$'s. 
Thus, if $k=2$, then $G_x \subseteq G_{y_1}$ -- contradicting $T$ being a clique tree. 
Now suppose $k \geq 3$ and consider $v \in G_x$ where (w.l.o.g.) $P_{v}$ contains $y_1$ and $y_2$. 
Since $G$ is connected, there must be $v' \in G_x \cap G_{y_3}$. 
Furthermore, (w.l.o.g.) $P_{v'}$  contains $y_1$ (otherwise, $P_v$ and $P_{v'}$ cross). 
Now, since $G_x$ and $G_{y_1}$ are maximal cliques, there is $v'' \in G_x \setminus G_{y_1}$. 
Notice that $P_{v''}$ must contain $y_2$ and $y_3$ in order for $P_{v''}$ to cross neither $P_{v}$ nor $P_{v'}$. 
Finally consider any $u \in G_x \setminus \{v,v',v''\}$. 
Notice that, in order for $P_u$ to not cross any of $P_{v}$, $P_{v'}$, or $P_{v''}$, it must contain at least two of $y_1,y_2,y_3$. 
In particular, if $k \geq 4$, then $G_x \cap G_{y_4} = \emptyset$ -- contradicting $G$ being connected. 
Thus, $k = 3$ (establishing (i)).  

Now, suppose that $y_1$ is not a terminal. By 1. and 2.(i), $y_1$ is either a junction with degree 3 or mixed with degree 2.

\textbf{Case 1:} \textit{$y_1$ is a junction with neighbors $x$, $z_1$, $z_2$.} 
Notice that each of $P_{v}$ and $P_{v'}$ must contain exactly one of $z_1$ or $z_2$. 
Moreover, w.l.o.g. they both must contain $z_1$ otherwise they will cross. 
However, since $y_1$ is a junction, we have vertices $w,w',w''$ such that  
$P_w \supseteq \{x,y_1,z_1\}$, $P_{w'} \supseteq \{x,y_1,z_2\}$ and $P_{w''} \supseteq \{z_1,y_1,z_2\}$. 
Moreover, both $P_w$ and $P_{w'}$ must contain either $y_2$ or $y_3$. 
Regardless of this choice, we end up with a crossing between either $P_{w'}$ and $P_v$ or $P_{w'}$ and $P_{v'}$. 
Thus, junctions cannot be neighbors.

\textbf{Case 2:} \textit{$y_1$ has degree 2 and is mixed.} 
Let $z$ be the neighbor of $y_1$ other than $x$ and let $w$ be a vertex of $G$ where $y_1$ is not a leaf of $P_w$, i.e., w.l.o.g. $P_w \supseteq \{z,y_1,x, y_2\}$. 
Notice that, $P_{v'}$ must also contain $z$ otherwise $P_{v'}$ and $P_w$ would cross. 
Similarly, since $P_{v'}$ now contains $z$, $P_v$ must also contain $z$ otherwise $P_v$ and $P_{v'}$ would cross. 
However, now a vertex $u \in G_{y_1} \setminus G_{z}$ must have $P_u = \{y_1\}$ but then $P_u$ crosses $P_w$. 
Thus, no neighbor of a junction is mixed. 
\end{claimproof}

\begin{claimproof}
\textbf{3.:} This follows immediately from 1. and 2.(i). 
\end{claimproof}
\end{proof}

\begin{observation}\label{obs:path-tree.nets.suns}
For an NC-path-tree graph $G$, let $(\{P_v : v\in V(G)\},T)$ be its clique NC-path-tree model. 
Let $x$ be a node of $T$ of degree at least three. 
\begin{enumerate}[noitemsep,topsep=0pt]
\item\label{prop:twin-free-junction} If $x$ is a junction, then $G$ contains a 3-sun. Also, if $G$ is twin-free, $|G_x|=3$. 
\item If $x$ is a terminal, then $G$ contains a net.
\end{enumerate}
\end{observation}
\begin{proof}
We establish the claimed properties in order as follows.

\begin{claimproof}
\textbf{1.:} As in the proof of Lemma~\ref{lem:nc-path-tree.nodes}.2.(i) a junction $x$ in $T$ has three neighbors $y_1,y_2,y_3$ and vertices $v,v',v'' \in G_x$ such that $P_v \supseteq \{y_1,x,y_2\}$, $P_{v'} \supseteq \{y_1,x,y_3\}$ and $P_{v''} \supseteq \{y_2,x,y_3\}$. 
Additionally, since $x,y_1,y_2,y_3$ are maximal cliques, there are vertices $u_1,u_2,u_3 \in V(G)$ such that $u_i \in G_x \setminus G_{y_i}$ for each $i \in \{1,2,3\}$. 
Moreover, all of these vertices are distinct due to their paths being incomparable. 
Thus, by considering the 3-sun and its clique tree model given in Fig.~\ref{fig:small-graphs}, it is now easy to see that $G[v,v',v'',u_1,u_2,u_3]$ is a 3-sun. 
Furthermore, since $y_1,y_2,y_3$ are terminals, the paths $P_v,P_{v'},P_{v''}$ are the only distinct paths which are possible for vertices in $G_x$. In other words, every vertex in $G_x \setminus \{v,v',v''\}$ is a twin of one of $v$, $v'$, or $v''$. 
\end{claimproof}

\begin{claimproof}
\textbf{2.:} Let $y_1,y_2,y_3$ be distinct neighbors of $x$. Since $G$ is connected and $x,y_1,y_2,y_3$ are maximal cliques, we have $v_i \in G_x \cap G_{y_i}$ and $u_i \in G_{y_i} \setminus G_x$ for each $i \in \{1,2,3\}$. 
The $v_i$'s are distinct since $x$ is a terminal, and the $u_i$'s are distinct since their paths are disjoint. 
Thus, by considering the net and its clique tree model given in Fig.~\ref{fig:small-graphs}, it is easy to see that $G[v_1,v_2,v_3,u_1,u_2,u_3]$ is a net.
\end{claimproof}
\end{proof}

\begin{corollary}
\label{cor:nc-path-tree.structure}
For an NC-path-tree graph $G$ and its clique NC-path-tree model $T$, the edges of $T$ uniquely partition into connected subtrees so that each subtree $T'$ has one of the following two types. 
\begin{enumerate}
\item\label{prop:structure.junction} $T'$ consists of the three edges incident to a junction $x$, i.e., $T'$ is the $K_{1,3}$ formed by $x$ together with its three neighbors $y_1, y_2, y_3$ (all of which are terminals). 
\item\label{prop:structure.mixed-path} $T'$ is a path where the two end nodes are terminals and each inner node (if there are any) has degree 2 and is mixed.
\end{enumerate}
\end{corollary}
\begin{proof}
This follows from Lemma~\ref{lem:nc-path-tree.nodes} and by simply partitioning the edges of $T$ into maximal connected sets delimited by the terminals of $T$. 
\end{proof}

\subsection{Restricted Host Trees}

Here we relate and characterize the classes of NC-path-d.tree, NC-path-r.tree, and NC-path-path graphs as stated in the next two theorems. 
While Theorem~\ref{thm:nc-path-path.fisc} (concerning NC-path-path graphs) is well known~\cite{Roberts1970,Wegner1967}, it also follows directly from our study of NC-path-tree graphs.

\begin{theorem}\label{thm:nc-path-r.tree-fisc}
A graph $G$ is (claw,3-sun)-free chordal if and only if it is NC-path-r.tree. Moreover, a graph has an NC-path-\textbf{d}.tree model if and only if it has a clique NC-path-\textbf{r}.tree. 
\end{theorem}
\begin{proof} \ %

$\Leftarrow$ It is known and easy to see that the 3-sun is not a path-d.tree graph \cite{ChaplickGLT2010}. 
Thus, NC-path-d.tree is a subclass of (3-sun)-free NC-path-tree = (claw,3-sun)-free chordal by Theorem~\ref{thm:nc-path-tree-fisc}. 

$\Rightarrow$ 
By Theorem~\ref{thm:nc-path-tree-fisc} and Observation~\ref{obs:path-tree.nets.suns}, for every (3-sun,claw)-free chordal graph $G$, there are no junctions in the clique NC-path-tree model $(\{P_v\}_{v \in V(G)},T)$ of $G$. 
Thus, since every node of $T$ with degree at least three is a terminal, rooting $T$ at any terminal results in an NC-path-r.tree model. 
\end{proof}

\begin{theorem}\label{thm:nc-path-path.fisc}
A graph $G$ is (claw,3-sun,net)-free chordal if and only if it is NC-path-path, i.e., proper interval. 
\end{theorem}
\begin{proof} \ %

$\Leftarrow$ It is known and easy to see that the net is not a path-path (interval) graph~\cite{LB1962}. 
Thus, NC-path-path is a subclass of (net)-free NC-path-r.tree = (claw,3-sun,net)-free chordal by Theorem~\ref{thm:nc-path-r.tree-fisc}. 

$\Rightarrow$ As in the proof of Theorem~\ref{thm:nc-path-r.tree-fisc}, we note that since $G$ is a (net,3-sun)-free NC-path-tree graph, by Observation~\ref{obs:path-tree.nets.suns}, its unique clique NC-path-tree model has maximum degree two. Thus, the host is a path. 
\end{proof}

\subsection{Recognition Algorithms}
\label{sec:nc-path.recog}

From our characterizations, there are straightforward polynomial-time certifying algorithms for the classes of NC-path-tree and NC-path-r.tree graphs. 
Specifically, since these classes are characterized as chordal graphs with an additional finite set of forbidden induced subgraphs, we can apply a linear time certifying algorithm for chordal graphs~\cite{RoseTL1976}, and then apply brute-force search for our additional forbidden induced subgraphs. 
If no forbidden induced subgraph is found, we can simply construct the unique clique tree of the given graph (e.g., using~\cite{GalinierHP1995}) and it will be an NC-path-tree (or NC-path-r.tree) model as needed to positively certify membership in our classes. 
However, we can do this more carefully and obtain linear time certifying algorithms as in the next theorem.
A direct consequence of our certifying algorithm is that one can determine whether a chordal graph contains an induced claw in linear time. 
As we mentioned before, this stands in contrast to the case of general graphs where the best deterministic algorithms run in time $O(\min\{n^{3.252},m^{(\omega+1)/2}\})$~\cite{EisenbrandG04}, and $O(m^{\frac{2\omega}{\omega+1}})$~\cite{KloksKM00}.

\begin{theorem}\label{thm:nc-path-tree-recog}
The classes NC-path-tree and NC-path-r.tree (= NC-path-d.tree) have linear-time certifying algorithms. In particular, one can certify the presence/absence of an induced claw in a chordal graph in linear time. 
\end{theorem}
\begin{proof}
Recall that the size $\sum_{v \in V(G)} |K_v|$ of a clique tree is $O(n+m)$ (we use this implicitly throughout the following).
First, we run a linear-time certifying algorithm for chordal graphs, e.g.,~\cite{RoseTL1976}. 
Then, we construct a clique tree $T$ in linear-time~\cite{GalinierHP1995}. 
We then annotate the clique tree to mark, for each vertex, for each maximal clique $K$ in $K_v$, if $K$ is a leaf or an internal node of the model of $v$. 
This annotation can be carried out in linear time because the total size of the clique tree is $O(n+m)$. 
If some vertex $v$ uses $\geq 3$ cliques as leaves, we produce a claw as in Claim~1 of the proof of Theorem~\ref{thm:nc-path-tree-fisc}. 
If there is a mixed node $x$ of degree $\geq 3$, then we proceed as in the proof of Lemma~\ref{lem:nc-path-tree.nodes}.1. 
This provides us with a pair of paths that cross in linear time. 
Then, proceeding as in Claim~2 in the proof of Theorem~\ref{thm:nc-path-tree-fisc}, we identify a claw.
Now all of the nodes of degree $\geq 3$ are either terminals or junctions, and we mark them as such. 
So, if there is a junction $x$ with degree $\geq 4$, we proceed as in Lemma~\ref{lem:nc-path-tree.nodes}.2.(i) to identify a pair of paths that cross and, as before, report a corresponding claw. 
Furthermore, if a junction $x$ neighbors a non-terminal $y$, we proceed as in Lemma~\ref{lem:nc-path-tree.nodes}.2.(i) to identify a pair of paths that cross and (again) a corresponding claw. 

Now, no crossing between two paths can involve a node of degree $\geq 3$. 
So, it remains just to ensure no crossings occur on a path between such nodes. 
In particular, since the neighbors of all junctions are terminals, such a crossing must occur on a path connecting two terminals (where all of the inner nodes are mixed, and, by Lemma~\ref{lem:nc-path-tree.nodes}, have degree two in $T$). 
Let $x_1, \ldots, x_k$ be such a path. 
Clearly, this path of cliques represents an interval graph. 
Moreover, we will find a pair of crossing paths on it precisely when this interval graph is not a proper interval graph. 
Conveniently, this problem is known to be solvable in linear time~\cite{DengHH1996}. 
However, to obtain linear time in total (when processing all such paths) we need to be a bit careful. 
Namely, rather than simply checking whether each $G[\bigcup_{i=1}^k G_{x_i}]$ is a proper interval graph, for each such path we create the following auxiliary graph $G'$. 

\paragraph{\textbf{The graph} $\mathbf{G'}$ built from a path $x_1, \ldots, x_k$ in $T$ where $x_1$ and $x_k$ are terminals and each $x_i$ ($i \in \{2, \ldots, k-1\}$) is mixed.}
The vertex set of $G'$ is $\{u_1,u_k\} \cup \bigcup_{i=2}^{k-1} G_{x_i}$. 
In $G'$, for each $i \in \{2, \ldots, k-1\}$, we make $G_{x_i}$ a clique. 
Also, we make $u_1$ adjacent to $G_{x_1} \cap G_{x_2}$ and $u_k$ is adjacent to $G_{x_{k-1}} \cap G_{x_k}$. 
In this way, the size of $G'$ can easily be seen as linear in the size of $G[\bigcup_{i=2}^{k-1} G_{x_i}]$. 
Moreover, since we only consider paths connecting terminals, each vertex and edge of $G$ is contained in at most one $G'$. 
Finally, observe that $G'$ is interval and is a proper interval graph if and only if  $G[\bigcup_{i=2}^{k-1} G_{x_i}]$ is as well. 

Thus, running the certifying recognition algorithm for proper interval graphs on $G'$ will provide a claw when $G'$ is not a proper interval graph, and such a claw is easily mapped back to a claw in $G$. 

This completes the case of NC-path-tree graphs. 
For NC-path-r.tree graphs, we additionally check if $T$ contains junctions and proceed as in Observation~\ref{obs:path-tree.nets.suns}.1. 
In particular, if there are no junctions, we have an NC-path-r.tree model, and if a junction is present, we easily report a 3-sun to certify that the graph is not an NC-path-r.tree graph as described in Observation~\ref{obs:path-tree.nets.suns}.1. 
\end{proof}

\section{Domination Problems}
\label{sec:mds}

A \emph{dominating set} in a graph $G$ is a subset $D$ of $V(G)$ such that every vertex is either in $D$ or adjacent to a vertex in $D$. 
In the \emph{minimum dominating set (MDS)} problem a graph $G$ is given and the goal is to determine a dominating set in $G$ with the fewest vertices. 
The MDS problem is \cNP-complete on PT graphs~\cite{BoothJ1982}, and split graphs~\cite{CorneilP1984}, and line graphs of planar graphs~\cite{YannakakisG1980} (which are of course claw-free). 

A dominating set in a graph is \emph{independent} when the subgraph it induces is edgeless. 
Interestingly, the \emph{minimum independent dominating set (MIDS)} problem (defined analgously to the MDS problem) can be solved on chordal graphs in linear time~\cite{Farber1982}. 
For NC-path-tree graphs, the size of an MIDS is the same as the size of an MDS as shown in the conference version of this paper~\cite{wg2019}.  
However, this is also true for claw-free graphs~\cite{AllanL78,FaudreeFR97} (which, due to our characterization, trivially form a superclass of the NC-path-tree graphs).
This implies the following theorem. 

\begin{theorem}\label{thm:mds->mids}
For any NC-path-tree graph $G$, there is an independent dominating set that is also a minimum dominating set.
Moreover, such an independent dominating set can be found in linear time.
\end{theorem}

So, we turn to another natural domination problem on NC-path-tree graphs. 
A dominating set in a graph is \emph{connected} when the subgraph it induces is connected.
In the \emph{minimum connected dominating set (MCDS)} problem, the input is a graph $G$, and the goal is to find a connected dominating set with the fewest vertices. 
The MCDS problem is \cNP-hard even on \emph{line graphs} of planar graphs of maximum degree four~\cite[Lemma~46]{Munaro17}\cite[Theorem 10.5]{HermelinMLW19} (a quite restricted subclass of claw-free graphs) but fixed-parameter tractable on claw-free graphs~\cite{HermelinMLW19}. (In fact, under the \emph{Exponential Time Hypothesis}~\cite{ImpagliazzoPZ01}, there is no constant $c$ such that there is a $2^{o(k)}n^c$ time algorithm to decide whether a line graph has a connected dominating set of size $k$~\cite[Corollary 10.9]{HermelinMLW19}.)
This problem is also \cNP-hard on split graphs but can be solved in polynomial time on strongly chordal graphs~\cite{WhiteFP85}. 
Note that the strongly chordal graphs include the NC-d.path-tree graphs, but do not include the NC-path-tree graphs (since, e.g., the 3-sun is an NC-path-tree graph but it is not strongly chordal~\cite{Farber1983}). 
Interestingly, it has also been shown~\cite[Corollary 4.3, Theorem 4.4]{WhiteFP85} that, for chordal graphs, the MCDS problem and the (cardinality) Steiner tree problem, defined next, are equivalent under linear time reductions.
 
For a graph $G$ and subset $X$ of $V(G)$, a \emph{Steiner tree (ST) of $X$} is a subtree of $G$ that includes $X$. 
In the ST problem, the input is a graph $G$ and a subset $X$ of $V(G)$, and the aim to find a ST of $X$ with the fewest vertices.  

We will now establish the following theorem and corollary regarding the MCDS and ST problems on NC-path-tree graphs (note that the corollary follows simply from the theorem and~\cite[Theorem 4.4]{WhiteFP85}).

\begin{theorem}
\label{thm:mcds}
For any connected NC-path-tree graph $G$, a minimum connected dominating set of $G$ can be produced in linear time.
\end{theorem}

\begin{corollary}
\label{cor:st}
For any NC-path-tree graph $G$ and subset $X$ of the vertices of a connected component of $G$, a minimum cardinality Steiner tree can be produced in linear time.
\end{corollary}

To establish Theorem~\ref{thm:mcds} (and Corollary~\ref{cor:st}), we design an algorithm based on the following lemma. 
This lemma characterizes every MCDS in an NC-path-tree graph via the terminals and junctions of its clique NC-path-tree model. 
Recall that, as formalized in Corollary~\ref{cor:nc-path-tree.structure},  by thinking of the terminals in the clique NC-path-tree model $T$ of any NC-path-tree graph $G$ as delimiters, the edges of $T$ partition into the following  two special types of subtrees. 
\begin{itemize}
\item A junction $x$ together with its three neighbors $y_1, y_2, y_3$ (all of which are terminals). 
\item A path $P$ where the two end nodes are terminals and each inner node (if there are any) has degree 2 and is mixed. (In the lemma below, we will also differentiate the case when the $P$ is a single edge).
\end{itemize}

\begin{lemma}
\label{lem:mcds-structure}
Let $(\{P_v\}_{v \in V(G)},T)$ be a clique NC-path-tree model a graph $G$ where $G$ is not a clique. 
A subset $D$ of $V(G)$ is an MCDS of $G$ if and only if properties 1--3 below are satisfied.
\begin{enumerate}

\item\label{prop:mcds-junction} 
For each junction $x$ in $T$, $D$ contains exactly two vertices $u$ and $v$ from $G_x$ so that $P_u \cup P_v$ includes the three (terminal) neighbors of $x$ (i.e., $u$ and $v$ are not twins). 

\item\label{prop:mcds:terminal-termainal} 
For each edge $xy$ in $T$ where both $x$ and $y$ are terminals, $D$ contains exactly one vertex of $G_x \cap G_y$. 

\item\label{prop:mcds-mixed-path} 
For each path $(z_1, \ldots, z_k)$ in $T$ where $k \geq 3$, both $z_1$ and $z_k$ are terminals, and each $z_i$ ($i \in \{2, \ldots, k-1\}$) has degree 2, 
we have that the subgraph of $G$ induced by $D \cap \bigcup_{i=2}^{k-1} G_{z_i})$ is a shortest path connecting each $u \in G_{z_1}\setminus G_{z_2}$ to each $v \in G_{z_k} \setminus G_{z_{k-1}}$. 
\end{enumerate}
\end{lemma}
\begin{proof}
The key to this proof is the next simple claim regarding connected dominating sets.

\smallskip

\claim{$\star$: A subset $S$ of $G$'s vertices is a connected dominating set if and only if for every edge $xy$ of $T$, $S$ contains at least one vertex of $G_x \cap G_y$. }
\begin{claimproof}
\noindent\textit{Proof of Claim~$\star$.}

$\Leftarrow$  
First, since $G$ is not a clique, $T$ contains at least one edge. 
Therefore, since every node $x$ of $T$ is incident to some edge, $S$ contains a vertex of $G_x$ for every node $x$ of $T$, i.e., $S$ dominates $G$. 
Second, we have that $\bigcup_{v \in S} P_v$ is connected (since it is equal to $T$). 
In particular, $G[S]$ is connected.

$\Rightarrow$ 
Suppose that there is an edge $xy$ of $T$ where for every $v \in S$, $xy$ does not belong to~$P_v$. 
Let $T'$ and $T''$ be the two subtrees of $T$ obtained by deleting the edge $x$ from~$T$. 
Now, since $S$ is dominating, it contains a vertex $u$ whose model (path) in $T$ is contained in~$T'$ and a vertex $v$ whose model (path) in $T$ is contained in~$T''$. 
However, every $(u,v)$-path in $G$ must contain a vertex of $G_x \cap G_y$: This contradicts $G[S]$ being connected. 
\end{claimproof}
 
\medskip  
\noindent 
The key consequence of Claim~$\star$ is that an MCDS is, equivalently, a smallest set $S$ of vertices where every edge of $T$ is included in the path of some vertex in $S$. 
In particular, to characterize the MCDSs, it suffices to independently consider each subtree of $T$ as in the edge-partition with respect to terminals stated in Corollary~\ref{cor:nc-path-tree.structure} (and, also as slightly more finely enumerated in the statement of this lemma). 
With this in mind, we proceed with the proof for each item of the enumeration separately. 

\medskip 

\begin{claimproof}
\textbf{1.:}  
Here we have to cover the three edges incident to a junction $x$. 
Clearly, doing so requires at least two paths arising from non-twin vertices. 
Moreover, the only vertices whose paths contain edges incident to $x$, are those in $G_x$. 
Thus, we must pick two non-twin vertices of $G_x$, and, since $x$ is a junction, by doing so we indeed obtain two paths that cover all three edges. 
\end{claimproof}

\begin{claimproof}
\textbf{2.:}  
Here, we just need to ensure the edge $xy$ is covered. 
Since, $G$ is connected, there must be at least one vertex whose path includes this edge. 
In particular, $G_x \cap G_y \neq \emptyset$ and it suffices to just take any such vertex. 
\end{claimproof}

\begin{claimproof}
\textbf{3.:} 
Finally, we arrive at the somewhat non-trivial case concerning a path $P = (z_1, \ldots, z_k)$ in $T$ 
where $z_1$ and $z_k$ are terminals, and each $z_i$ ($i \in \{2, \ldots, k-1\}$) is mixed (and as such has degree 2). 

Let $u$ be a vertex in $G_{z_1} \setminus G_{z_2}$, and let $v$ be a vertex in $G_{z_k} \setminus G_{z_{k-1}}$. 
In claims~(a) and~(b) below, we establish that (a) for any induced $(u,v)$-path in $G$, the models (paths) of the inner vertices cover the edges of $P$ (in $T$); and, (b) that in any MCDS, the vertices whose models (paths) include edges of $P$ constitute the inner vertices of an induced $(u,v)$-path. 
Together, these claims indeed imply Property~\ref{prop:structure.mixed-path} of this lemma since shortest paths are the smallest induced paths.

\claim{(a) For any induced $(u,v)$-path $Q = (u, w_1, \ldots, w_\ell, v)$ in $G$, the models (paths) of the inner vertices $(w_i, i \in \{1, \ldots, k\})$ cover precisely the edges of $P$ (in $T$).}

First, observe that $u$ and $v$ are not adjacent, i.e., $\ell \geq 1$ and $Q$ contains inner vertices. 
Second, observe that, for any inner vertex $w_i$ ($i \in \{1, \ldots, \ell\}$ of $Q$, $P_{w_i}$ is a subpath of~$P$ since $Q$ is an induced path and $z_1$ and $z_k$ are terminals.  
Finally, similarly to the proof of~$\Rightarrow$ for Claim~$\star$, since $Q$ is connected and includes $u$ and $v$ where $P_u$ and $P_v$ are separated by the path $P$ in $T$, we have that $\bigcup_{i=1}^{k} P_{w_i} \supseteq P$; thus, $\bigcup_{i=1}^{k} P_{w_i} = P$, completing the proof of this claim.

\claim{(b) In any MCDS $D$ of $G$, the vertices whose models (paths) include edges of $P$ constitute the inner vertices of an induced $(u,v)$-path.}

Let $D_P$ be the subset of $D$ where $w \in D_P$ if and only if $P_w$ contains an edge of $P$. 
By Claim~$\star$, $D_P \neq \emptyset$ and $P \supseteq \bigcup_{w \in D_P} P_w$. 
Moreover, since $z_1$ and $z_k$ are terminals, the union $\bigcup_{w \in D_P} P_w$ is contained in $P$, i.e., $P = \bigcup_{w \in D_P} P_w$.
In particular, $D_P$ induces a proper interval subgraph of $G$. 

We now show that $D_P$ induces a path $w_1, \ldots, w_\ell$ in $G$ such that $w_1$ is adjacent to $v$ and $w_\ell$ is adjacent to $u$. 

We first establish the adjacency to $u$ and $v$. 
By Claim~$\star$, $D_P \cap G_{z_1} \cap G_{z_2} \neq \emptyset$, and we pick $w_1$ as any vertex in $D_P \cap G_{z_1} \cap G_{z_2}$. 
Since $w_1 \in G_{z_1}$, we indeed have that $uw_1$ is an edge. 
Notice that, if there is a vertex $w \in D_P \cap G_{z_1} \cap G_{z_2}$ such that $w \neq w_1$, then either $P_w \supseteq P_{w_1}$ or $P_{w_1} \supseteq P_w$, i.e., this would contradict the fact that $D$ is an MCDS.
Thus, $w_1$ is the only vertex of $D$ in $G_{z_1} \cap G_{z_2}$. 
(Symmetrically, we have a vertex $w_\ell$ adjacent to $v$ such that $\{w_\ell\} = D_P \cap G_{z_{k-1}} \cap G_{z_k}$.)

To establish that $D_P$ really induces a path, we remark that it suffices to show that $G[D_P]$ is triangle-free. 
In particular, it is known~\cite{Eckhoff93} that a triangle-free interval graph is a caterpillar. 
Thus, since $G[D_P]$ is a proper interval graph (and as such claw-free---recall Theorem~\ref{thm:nc-path-path.fisc}), if $G[D_P]$ is triangle-free, then it is indeed a path. 
Moreover, the $(w_1,w_\ell)$ subpath of $G[D_P]$ constitutes a set of inner vertices of an induced $(u,v)$-path in $G$, and thus, by Claim~(a) and the minimality of $D$, $G[D_P]$ is precisely this subpath. 

We now establish that $G[D_P]$ is triangle-free to show that it is indeed a path, completing the proof of Claim~(b). 
Suppose (for a contradiction) that $G[D_P]$ contains a triangle $w,w',w''$. 
By the Helly property of subtrees of a tree, we have that there is $z_j$ such that $z_j \in P_w \cap P_{w'} \cap P_{w''}$. 
However, since each of $P_w$, $P_{w'}$, and $P_{w''}$ is a subpath of $P$, without loss of generality, we have that $P_w \subseteq P_{w'} \cup P_{w''}$. 
This contradicts the minimality of $D$, and establishes that $G[D_P]$ is indeed triangle-free. 
Therefore, $G[D_P]$ is indeed a path and we have established Claim~(b). 

\noindent Finally, as remarked above, combining Claims (a) and (b) establishes Property~\ref{prop:mcds-mixed-path}. 
\end{claimproof}

\end{proof}

Based on the above lemma we will now prove Theorem~\ref{thm:mcds}, establishing our linear time algorithm for the MCDS problem. 

\begin{proof}[Proof of Theorem~\ref{thm:mcds}]
In essence, this is just describing how to efficiently determine the vertices of an MCDS as described by the properties established in Lemma~\ref{lem:mcds-structure}. 
First, as in our certifying recognition algorithm (in the proof of Theorem~\ref{thm:nc-path-tree-recog}), we construct the clique NC-path-tree model $(\{P_v\}_{v \in V(G)}, T)$ of $G$ and mark each node as mixed, terminal, or junction. 
This allows us to partition $T$ according to its terminals as in Corollary~\ref{cor:nc-path-tree.structure}. 

Now, as justified by Property~\ref{prop:mcds-junction}, for each junction, we simply pick any two non-twin vertices. 
Let $D_1$ be this set of vertices. 

Similarly, as justified by Property~\ref{prop:mcds:terminal-termainal}, for each edge connecting two terminals, we simply pick any vertex whose path contains this edge (actually, the path of any such vertex will be precisely this edge). 
Let $D_2$ be this set of vertices. 

As justified by Property~\ref{prop:mcds-mixed-path}, for each path of mixed nodes connecting two terminals in $T$, we will compute an appropriate shortest path in $G$. 
Of course, here, to obtain a linear running time, we have to be a bit careful. 
Let $(z_1, \ldots, z_k)$ be a path in $T$ where $z_1$ and $z_k$ are terminals and for each $i \in \{2, \ldots, k-1\}$, $z_i$ is mixed (and, as such, has degree 2).
Here, we again construct the auxiliary graph $G'$ (as described in the proof of Theorem~\ref{thm:nc-path-tree-recog}) for this path $(z_1, \ldots, z_k)$, and compute a shortest path between the special vertices $u_1$ and $u_k$. 
Note that, since $G'$ is an interval graph, such a shortest path can be computed in linear time~\cite{AtallahCL95}. 
After doing so, we simply keep the inner vertices of such a path for our MCDS.
Moreover, as remarked before, the total size of all of these $G'$ graphs is linear in the size of $G$, thus we can compute a shortest path for each such $G'$ graph in linear time in total. 
This gives us the set $D_3$ consisting of the inner vertices from this collection of paths. 

Finally, we output the set $D_1 \cup D_2 \cup D_3$ as our MCDS. 
\end{proof}

\section{Hamiltonian Cycles and Minimum Leaf Spanning Trees}
\label{sec:ham}

As mentioned earlier, the HC and HP problems are \cNP-complete on DPT graphs and split graphs.
They are also \cNP-complete on line graphs of biparite graphs, i.e., (claw, diamond, odd-hole)-free graphs~\cite{LaiW93}, where the \emph{diamond} is the graph obtained by removing one edge from $K_4$.
In contrast, we show that, like proper interval graphs~\cite{Bertossi83}, 2-connectivity suffices for Hamiltonicity in NC-path-tree graphs, but additionally, every \emph{tracing} of a clique NC-path-tree model provides a distinct HC of its graph.
We similarly characterize the presence of an~HP via an obvious necessary condition in Theorem~\ref{thm:hamilton-path} below. 
This characterization of HPs directly allows us to characterize the number of leaves in a minimum-leaf spanning tree, see Corollary~\ref{cor:min-leaf-spanning-tree}, and ultimately provide a linear time algorithm for the minimum-leaf spanning tree problem on NC-path-tree graphs. 

\begin{theorem}\label{thm:hamiltonicity}
An NC-path-tree graph $G$ has a Hamiltonian cycle if and only if it is 2-connected and has at least three vertices.
Also, for each plane layout of $G$'s clique NC-path-tree model $T$, a distinct a Hamiltonian cycle of $G$ can be constructed in linear time. 
\end{theorem}
\begin{proof}
We build on the fact that 2-connected proper interval graphs are not only Hamiltonian but have an HC with quite special structure, established in~\cite{Bertossi83}, and described as follows.
Consider a proper interval graph $G$. 
Let $x_1, \ldots, x_k$ be the maximal cliques $G$ ordered according to the clique NC-path-path model of $G$. 
Further, let $u_1$ be a vertex of $G_{x_1} \setminus G_{x_2}$ and let $u_k$ be a vertex of $G_{x_k} \setminus G_{x_{k-1}}$. 
When $G$ is 2-connected there are internally disjoint $(u_1,u_k)$-paths $Q_1$ and $Q_2$ such that every vertex of $G$ belongs to either $Q_1$ or $Q_2$. 
Importantly for our claimed time bound is that these two paths can be found in linear time~\cite{Bertossi83}.
In essence, we will see (through an auxiliary multigraph $X$ constructed below) that such paths also occur in 2-connected NC-path-tree graphs by considering the proper interval graphs occurring between terminals. 

Now consider a 2-connected NC-path-tree graph $G$ and its clique NC-path-tree model $T$. 
Recall that, as we noted when designing our certifying algorithm for NC-path-tree graphs, for a path $x_1, \ldots, x_k$ in $T$ where $x_1$ and $x_k$ are terminals and each inner node is mixed (and consequently of degree 2), the graph $G[\bigcup_{i=1}^k G_{x_i}]$ is a proper interval graph. 
Moreover, since $G$ is 2-connected, each such subgraph is also 2-connected. 
Additionally, the graph $G'$ created from $G[\bigcup_{i=1}^k G_{x_i}]$ as before is also 2-connected. 
However, there is one special case where we use a slightly different auxiliary graph (otherwise we simply use the $G'$ defined before). 
When $k=2$, the graph $G'$ is the clique $G_{x_1} \cap G_{x_k}$ together with new vertices $u_1$ and $u_k$ where $N(u_1) = N(u_k) = G_{x_1} \cap G_{x_k}$. 
Now, it is easy to see that each such graph $G'$ is 2-connected and proper interval, and since $u_1$ and $u_k$ are not adjacent, we have two non-empty disjoint paths that both start with a vertex of $G_{x_1} \cap G_{x_2}$, and end with a vertex of $G_{x_{k-1}} \cap G_{x_k}$. 
Moreover, as remarked above these two paths can be built in linear time. 

We now consider the case when a neighbor $y$ of $x$ is a junction before completing our construction of the HC. 
Let the other two neighbors of the junction $y$ be $x'$ and $x''$. 
Due to the fact that $x,x',x''$ are all terminals, the vertices of $G_y$ form three equivalence classes $A,A',A''$ of twins, where:
\begin{itemize}
\item each vertex in $A$ is represented by the path $x,y,x'$,
\item each vertex in $A'$ is represented by the path $x',y,x''$, and 
\item each vertex in $A''$ is represented by the path $x'',y,x$. 
\end{itemize}
Namely, using $A, A', A''$ we can ``traverse'' $T$ from $x$ to $x'$, from $x'$ to $x''$, and from $x''$ back to $x$. 
Due to the simple structure here, and the fact that partitioning into equivalence classes of twins is a linear time task, it is easy to construct the three steps of such a  traversal ``around'' all junctions in linear time in total.

Based on the above observations, we can now build our HCs. 
The intuition here is to consider the tree $T$ to be drawn crossing-free in the plane, and trace the outline of $T$ terminal-to-terminal by using the paths guaranteed by the above arguments. 
We will encode the family of all such traces by a multigraph $M$ formed on the terminals of $T$ where each Eulerian tour of $X$ will correspond to a distinct HC of $G$.
Namely, for each terminal $x$, and each neighbor $y$ of $x$ in $T$:
\begin{itemize}[noitemsep,topsep=0pt]
\item if $y$ is a terminal, then in $M$, $x$ and $y$  are connected by two edges (representing the two paths present in the corresponding $G'$).
\item if $y$ is a mixed node and $z$ is the terminal so that $y$ occurs on the $(x,z)$-path in $T$, then, in $M$, $x$ and $z$ are connected by two edges (representing the two paths present in the corresponding $G'$).
\item if $y$ is a junction and $x'$ and $x''$ are its two other neighbors, then in $M$, we have the edges $xx'$ and $xx''$.  
\item finally, if $G_x$ contains vertices that do not belong to any other $G_{x'}$ (e.g., when $x$ is a leaf of $T$), we also add a self-loop on $x$ in $M$ and map to this self-loop the vertices of $G_x \setminus (\bigcup_{x' \in N(x)}G_{x'})$. 
\end{itemize} 
We note the following properties of $M$ to complete the proof. 
The edges of $M$ partition the vertices of $G$ and each edge $xy$ corresponds to a path in $G$ where one end vertex  belongs $G_x$ and the other end vertex belongs to $G_y$. 
Furthermore, $X$ is Eulerian, each Eulerian cycle $C$ provides an HC, and $C$ describes a plane layout of $T$, i.e., a cyclic order of the edges around each node of $T$ so that $C$ traces the outline of this plane layout of $T$. 
Note that, each such plane layout will often arise from multiple Eulerian cycles in $M$, but no two distinct layouts arise from the same cycle. 
\end{proof}

We now turn to HPs, and ultimately to minimum-leaf spanning trees. 
Note that, an \emph{$\ell$-leaf spanning tree} of a graph $G$ is simply a spanning tree of $G$ with exactly $\ell$ leaves, and a \emph{minimum-leaf spanning tree} is a spanning tree having the fewest leaves. 
Clearly, checking for an HP is a special case of finding a minimum-leaf spanning tree. 
A natural lower bound on the number of leaves in a minimum-leaf spanning tree comes from looking at the block-cutpoint tree (defined next). 

The \emph{block-cutpoint tree} $BC(G)$ of a graph $G$ contains a node for each cut-vertex of $G$, a node for each maximal 2-connected subgraph (\emph{block}) of $G$, and its edge set is $\{ cB$ : $c$ is a cut-vertex, and $B$ is a block of $G$ containing $c$\}. 
It is well-known that $BC(G)$ can be computed in linear time~\cite{Hopcroft:1973}, and is indeed a tree. 
Clearly, if $G$ has an HP, $BC(G)$ is a path.
In the next theorem, we show the converse is also true in NC-path-tree graphs, and further below we generalize this to $\ell$-leaf spanning trees. 
The main idea is to observe where the cut-vertices occur in the model and then reuse our Eulerian structure $M$ from the previous proof. 
More generally, if $G$ has a spanning tree with at most $\ell$ leaves, then $BC(G)$ also can have at most $\ell$ leaves. 
Here, we observe that once we have the characterization for the presence of an HP, the converse of this easily follows, see Lemma~\ref{lem:hp-to-min-leaf}. 
In particular, it holds for NC-path-tree graphs, see Corollary~\ref{cor:min-leaf-spanning-tree}. 

\begin{theorem}\label{thm:hamilton-path}
An NC-path-tree graph $G$ contains a Hamiltonian path if and only if its block-cutpoint tree is a path. 
Moreover, when the block-cutpoint tree of $G$ is a path, a Hamiltonian path can be produced in linear time. 
\end{theorem}
\begin{proof}%
As noted above, it suffices to prove the $\Leftarrow$ direction. 
Let $G$ be an NC-path-tree graph and let $(\{P_v\}_{v \in V(G)},T)$ be its clique NC-path-tree model. 
Recall that, by Theorem~\ref{thm:hamiltonicity}, if $G$ has no cut-vertices, it has an HC, and thus also an HP. 
Therefore, we suppose $G$ contains a cut-vertex $v$. 
Note that $P_v$ must contain an edge $xy$ such that, for every vertex $u$ distinct from $v$, $xy$ is not an edge of $P_u$ (otherwise, $v$ would not be a cut-vertex). 
The next claim is the key to the proof.
\claim{: $P_v$ is precisely the edge $xy$ and both $x$ and $y$ are terminals.}
\begin{claimproof}
Note that $P_v$ cannot contain any junctions since the vertices whose paths use junctions cannot be cut-vertices (recall that, by Observation~\ref{obs:path-tree.nets.suns}, if $P_v$ contains a junction, $v$ is a central vertex of a 3-sun such that the two other central vertices are adjacent and dominate $N(v)$). 
Suppose that $y$ is not an end-node of $P_v$, and let $z$ be the neighbor of $y$ distinct from $x$ (note: $y$ is mixed and as such has degree 2 by Lemma~\ref{lem:nc-path-tree.nodes}.1). 
Now, since $G_z$ and $G_y$ are maximal cliques, we have a vertex $u \in G_y \setminus G_z$, but $u$ cannot belong to $G_x$ since $u \neq v$ and $P_v$ is the only path containing $xy$. Thus, $P_u = y$ and $P_u$ and $P_v$ cross. 
Furthermore, since $P_v$ is the only path that uses $xy$, both $x$ and $y$ must be terminals. 
\end{claimproof}

Note that, since $BC(G)$ is a path (and $G$ is not 2-connected), it consists of the two \emph{end} blocks (containing a single cut-vertex each) and (possibly) some \emph{inner} blocks containing exactly two cut-vertices each. 
Clearly, an HP must consist of one path in each block where, in the two end blocks, the cut-vertex is an end vertex of the path, and, in each inner block, the two cut-vertices are the two end vertices of the path. 

Let $B$ be an end block of $G$, i.e., a leaf of the block-cutpoint tree $BC(G)$ which contains one cut-vertex $v$.
Since $B$ is a 2-connected induced subgraph of $G$, $B$ has a Hamiltonian cycle $C_B$ by Theorem~\ref{thm:hamiltonicity}. 
So, to obtain a Hamiltonian path that ends at $v$, we just delete one edge incident to $v$ from $C_B$. 

To complete the proof, we will now argue that each inner block $B$ of $G$ containing two cut-vertices $v$ and $v'$ has a Hamiltonian path that connects $v$ to $v'$ within $B$. 

By the claim above, in the clique NC-path-tree model $(\{P^B_v\}_{v \in V(B)}, T^B)$ of $B$, each of $P^B_{v}$ and $P^B_{v'}$ is a single terminal node. Let these nodes be $x$ and $x'$ (respectively) in $T^B$.
Consider the path $x = x_1, x_2, \ldots, x_k = x'$ in $T^B$. 
Now, consider the Eulerian multigraph $M$ as in the proof of Theorem~\ref{thm:hamiltonicity}. 
Note that, since $x_1$ is a terminal, we can use $M$ to construct a path $Q_1$ that starts with $v$ and ends with a vertex of $G_{x_1} \setminus G_{x_2}$ and visits precisely the vertices in the connected component of $B \setminus (G_{x_1} \cap G_{x_2}$ that contains $v$. 
The path $Q_k$ is defined analogously.
Similarly, for each terminal $x_i$ $(i \in \{2, \ldots, k-1\})$, we can use $M$ to craft a path $Q_i$ that visits precisely the vertices whose paths occur strictly within the subtree of $T^B \setminus \{x_{i-1}, x_{i+1}\}$ that contains $x_i$. 
Moreover, this path will start and end with vertices whose paths contain $x_i$.
When $x_i$ is a junction, let $x'_i$ be the terminal that neighbors $x_i$ and is distinct from $x_{i-1}$ and $x_{i+1}$. 
Similarly to the case of $x_1$, we note that there is a path $Q'_i$ that visits all the vertices $U_i$ ``hanging below'' $x'_i$ and starts and ends with a vertex of $x'_i$. 
Additionally, due to the three equivalence classes of twins whose paths contain the junction $x_i$, we can extend this path $Q'_i$ to a path $Q_{i-1,i+1}$ that starts in a vertex of $G_{x_{i-1}} \cap G_{x_i}$, ends in a vertex of $G_{x_i} \cap G_{x_{i+1}}$, and visits every vertex of $U_i \cup G_{x_i}$.

Finally, consider two terminals $x_i$ and $x_j$ ($i < j$) where, for each $l \in \{i+1,\ldots, j-1\}$, $x_l$ is mixed. 
As in the proof of Theorem~\ref{thm:hamiltonicity}, we again consider the auxiliary graph $G'$ corresponding to this path. 
Here, we instead need a Hamiltonian path in $G'$ that starts and ends in our special vertices $u_1$ and $u_k$. 
Fortunately, it is known~\cite{Bertossi83}, that such a path does exist and actually only requires that $G'$ is connected. 
Namely, we have the path $Q_{i,j}$ which starts in a vertex of $G_{x_i} \cap G_{x_{i+1}}$, ends in a vertex of $G_{x_{k-1}} \cap G_{x_k}$, and visits every vertex of $\bigcup_{l=i+1}^{k-1}G_{x_l}$. 

Thus, to form a desired Hamiltonian path $Q_B$ of $B$ that starts with $v$ and ends with $v'$, we simply concatenate the paths $Q_1, Q_{1,i_1}, Q_{i_1}, Q_{i_1,i_2}, Q_{i_2}, \ldots, Q_{i_t}, Q_{i_t,k}, Q_k$ where $x_{i_1}, \ldots, x_{i_t}$ are the terminals that occur between $x_1$ and $x_k$. 
In particular, by forming such a Hamiltonian path $Q_B$ for each inner block $B$ of $G$, we are done.

We conclude by remarking that the construction of an HP here can be completed in linear time. 
In particular, it suffices to describe how to obtain linear time on each block separately. 
For each end block $B$, we simply invoke the HC algorithm (leading to linear time in the the size of $B$). 
For each inner block, the paths $Q_1, Q_{i_1}, \ldots, Q_{i_t}$, and $Q_k$ can similarly be constructed by invoking the HC algorithm, leading to linear time in the size of $B$ in total. 
Each of the other paths $Q_{1, i_1}, Q_{i_1, i_2}, \ldots, Q_{i+{t-1}, i_t},$ and $Q_{i_t, k}$ can also be constructed in linear time by a simple greedy algorithm~\cite{Bertossi83}. 
Thus since these paths, which we concatenate in order to make $Q_B$, are constructed from edge-disjoint induced subgraphs of $B$, the total time to construct $Q_B$ is also linear in the size of $B$.
\end{proof}

We now show how Theorem~\ref{thm:hamilton-path} can be generalized to minimum leaf spanning trees. 
While we expect that the following straightforward lemma has been observed before, we could not find an explicit proof of it, and so we include it here.

\begin{lemma}
\label{lem:hp-to-min-leaf}
For any graph class $\mathcal{G}$ closed under taking induced subgraphs, if every graph $G\in \mathcal{G}$ whose block-cutpoint tree is a path has a Hamiltonian path, then every graph $G \in \mathcal{G}$ whose block-cutpoint tree has $\ell$ leaves ($\ell \geq 2$) has a spanning tree with exactly $\ell$ leaves. 
\end{lemma}
\begin{proof}
We proceed by induction on the number of leaves in $BC(G)$. In the base case $BC(G)$ has two leaves, and the result follows trivially. 
So, suppose $BC(G)$ has $\ell \geq 3$ leaves. 
Let $Q$ be a path in $BC(G)$ that starts and ends in distinct leaf-blocks $B_1$ and $B_2$ where $B_1$ and $B_2$ share a cut vertex $v$, i.e., $Q$ is the path $(B_1,v,B_2)$. 
Further, let $G_Q$ be the subgraph of $G$ induced by the vertices occurring in the blocks on this path, and let $G'$ be the graph obtained by deleting every vertex of $B_1$ except $v$ from $G$. 
Now, by induction $G_Q$ has a Hamiltonian path $P$ and $G'$ has a spanning tree $T'$ with $\ell-1$ leaves. 
Moreover, since $v$ is a cut-vertex of $G_Q$, the path $P$ contains a subpath $P_1$ whose vertices are precisely the vertices of $B_1$ so that the vertex $v$ is an end vertex of $P_1$. 
Therefore, by gluing the path $P_1$ to $T'$ by identifying the occurrence of $v$ in both, we obtain a spanning tree $T$ of $G$ with precisely $\ell$ leaves. 
\end{proof}

\begin{corollary}
\label{cor:min-leaf-spanning-tree}
For any NC-path-tree graph $G$ that is not 2-connected (i.e., containing at least one cut-vertex), the number of leaves in a minimum-leaf spanning tree of $G$ is $\ell$ if and only if its block-cutpoint tree has exactly $\ell$ leaves. 
\end{corollary}
\begin{proof} 
$\Rightarrow$ Clearly, when the block-cutpoint tree has more than $\ell$ leaves, $G$ cannot have an $\ell$-leaf spanning tree. 

$\Leftarrow$ This follows directly from Theorem~\ref{thm:hamilton-path} and Lemma~\ref{lem:hp-to-min-leaf}.
\end{proof}

\section{Concluding Remarks}

In this paper we have studied intersection graph classes of non-crossing paths in trees. 
We have provided forbidden induced subgraph characterizations and recognition algorithms for the natural classes of such graphs. 
We have further studied and provided efficient algorithms for variations of domination and Hamiltonicity problems on intersection graphs of non-crossing paths in a tree. 

It might be interesting to investigate further algorithmic questions on this class that similarly have efficient algorithms on proper interval graphs, but are \cNP-hard on chordal graphs. 
A few problems in this context include: role assignment (aka locally surjective homomorphism) testing~\cite{Heggernes2012}, the simple max-cut problem~\cite{BodlaenderJ00} (here, the problem is still open even for proper interval graphs, see~\cite{BoyaciES2020-abs-2006-03856}), and the minimum outer-connected dominating set problem~\cite{KeilP13}.

Regarding further NC classes of graphs, a natural next step would be to study the NC-tree-tree graphs. 
But, as we mentioned before, it is not safe to simply work with clique trees in this case as the claw requires the use of a non-clique tree model. 
We conjecture that the NC-tree-tree graphs can be characterized as chordal graphs avoiding a finite set of forbidden induced subgraphs.
It would also be interesting to see if similar algorithmic results on domination and Hamiltonicity problems can be obtained on this class. 

Other host domains have been considered in the literature. 
Notice that, similar to proper interval graphs being NC-path-path graphs, the proper circular arc graphs are precisely the NC-path-cycle graphs. 
A simple host graph class that generalizes both trees and cycles is that of \emph{cacti}. 
A \emph{cactus} is a connected graph in which every 2-connected component is a single vertex, a single edge, or a chordless cycle.
The intersection graphs of 
subtrees of a cactus were studied by Gavril~\cite{Gavril1996}. 
So, one might consider the NC-path/tree/cactus-cactus graphs. 

Finally, an alternative view of host domains has been considered quite recently through the notion of \emph{$H$-graphs}~\cite{ChaplickFGKZ19,Chaplick0VZ17,ChaplickZ17,FominGR20}, i.e., for a fixed graph $H$, a graph $G$ is an $H$-graph when it is an intersection graph of connected subgraphs of a subdivision of $H$. 
Here, interval graphs are the $K_2$-graphs and circular-arc graphs are the $K_3$-graphs. 
While there is a natural notion of proper $H$-graphs~\cite{ChaplickFGKZ19} (which indeed restrict $H$-graphs for every $H$), the more restrictive non-crossing $H$-graphs might have a nicer structure and lead to easier (and faster) algorithms.

\bibliography{claw-free-intersection-graphs}{}
\bibliographystyle{splncs04}

\appendix

\newpage 

\end{document}